\documentclass{article}

\usepackage[T1,T2A]{fontenc}
\usepackage[utf8]{inputenc}
\usepackage[russian,spanish,english]{babel}

\usepackage[letterpaper,top=2cm,bottom=2cm,left=2.5cm,right=2.5cm,marginparwidth=1.75cm]{geometry}

\usepackage{amsmath}
\usepackage{amsthm}
\usepackage{amsfonts}
\usepackage{amssymb}
\usepackage{graphicx}
\usepackage[colorlinks=true, allcolors=blue]{hyperref}
\usepackage{authblk}
\newtheorem{mythm}{Theorem}

\newtheorem{mydef}{Definition}
\newtheorem{myexample}{Example}
\newtheorem{mylemma}{Lemma}
\newtheorem{myremark}{Remark}
\numberwithin{equation}{section}

\DeclareMathOperator{\diag}{diag}

\DeclareMathOperator{\myspan}{span}
\newcommand*\dif{\mathop{}\!\mathrm{d}}

\title{Dilation theorem via Schr\"odingerisation, with applications to the quantum simulation of differential equations}

\author[1]{Junpeng Hu\thanks{hjp3268@sjtu.edu.cn}}
\author[1,2,4]{Shi Jin\thanks{shijin-m@sjtu.edu.cn}}
\author[2,3,4]{Nana Liu\thanks{nana.liu@quantumlah.org}}
\author[1,2]{Lei Zhang\thanks{lzhang2012@sjtu.edu.cn}}
\affil[1]{School of Mathematical Sciences, Shanghai Jiao Tong University, Shanghai, 200240, P. R. China}
\affil[2]{Institute of Natural Sciences, MOE-LSC, Shanghai Jiao Tong University, Shanghai, 200240, P. R. China}
\affil[3]{University of Michigan-Shanghai Jiao Tong University Joint Institute, Shanghai, China}
\affil[4]{Shanghai Artificial Intelligence Laboratory, Shanghai, China}

\begin{document}
\maketitle

 \begin{abstract}
Nagy's unitary dilation theorem in operator theory asserts the possibility of dilating a contraction into a unitary operator. When used in quantum computing, its practical implementation primarily relies on block-encoding techniques, based on finite-dimensional scenarios. In this study, we delve into the recently devised Schr\"odingerisation approach and demonstrate its viability as an alternative  dilation technique. This approach is applicable to operators in the form of $V(t)=\exp(-At)$, which arises in wide-ranging applications, particularly in solving linear ordinary and partial differential equations. Importantly, the Schr\"odingerisation approach is adaptable to both finite and infinite-dimensional cases, in both countable and uncountable domains. For quantum systems lying in infinite dimensional Hilbert space, the dilation involves adding a single infinite dimensional mode, and this is the continuous-variable version of the Schr\"odingerisation procedure which makes it suitable for analog quantum computing.  Furthermore, by discretising continuous variables, the Schr\"odingerisation method can also be effectively employed in finite-dimensional scenarios suitable for 
qubit-based quantum computing. \end{abstract}

\tableofcontents

\section{Introduction}
\label{sec:introduction}

Quantum computing has gained significant attention in recent years due to its potential to provide (up to) exponential speedups compared to classical computing methods. One key aspect of this is quantum simulation, where the quantum system evolves via Schr\"odinger's equation and its description involves the use of unitary operators which maps one pure quantum state to another pure quantum state. Thus quantum devices are particularly suited to the simulation of unitary dynamics. On the other hand, for non-unitary dynamics, such as open quantum dynamics of quantum systems interacting with an environment, or more general ordinary and partial differential equations, and general linear algebra problems, one needs to find a way to "unitarise" it before it can be implemented in a quantum algorithm. 


The dilation theory offers a significant mathematical insight by showing how objects in a broader category, such as contractions, are intricately connected to better-behaved objects within a narrower category, such as unitaries or isometries. Sz.-Nagy’s unitary dilation theorem states that a contraction on a Hilbert space can be dilated into a unitary on a larger space, which is well understood by the spectral theorem and the dilation is unique up to an isomorphism \cite{nagy2010harmonic}. This powerful concept finds application in diverse areas. The minimal unitary dilation of a single contraction can serve as the basis for the development of operator theory for non-normal operators. It can also be used to define a refined functional calculus on contractions or analyze one-parameter semigroups of operators, and has led to considerable progress in the study of invariant subspaces \cite{shalit2021dilation}. In \cite{yan1977chinese}, Yan proved the conditions for one-parameter operator semigroups to have a cyclic unitary dilation. Shamovich and Vinnikov gave a construction of a dilation of the multi-parameter semigroup of commuting dissipative operators in \cite{shamovich2017dilations}.

The dilation theory provides a way to address non-unitary dynamics. For the purpose of creating a quantum circuit, the Sz.-Nagy dilation theorem allows one to simulate the effect of any non-unitary operator by a unitary quantum gate, because every bounded operator can be made into a contraction which has a unitary dilation. However, due to the large increase of the dimension of the Hilbert space, the computational costs as well the complexity of implementations required by different applications of the dilation theorem need to be specified for actual applications on a quantum computing device. 

The physical implication of the dilation theory is that a physical system that does not evolve under unitary dynamics can be enlarged where the enlarged system as a whole now evolves under unitary dynamics. For a finite $N$-dimensional system, described by an $N$-vector, a construction to demonstrate the dilation theorem requires a minimal increase to a $2N$-dimensional system. This $2N$-dimensional system can be described by a quantum state consisting of $\log_2(2N)=1+\log_2(N)$ quantum bits, or qubits. A single qubit corresponds to a two-level system. A physical system can also be infinite dimensional, for example a laser beam. Quantum infinite-dimensional systems are instead described by quantum modes, or qumodes, sometimes also called continuous-variable quantum systems \cite{braunstein2005quantum}. However, the current dilation theorem cannot be constructively applied for these infinite dimensional systems. In this paper, we show a simple explicit construction for these infinite dimensional systems where the inclusion of a single extra qumode is all that is necessary. This is in fact equivalent to adding a single spatial dimension to the partial differential equation that the infinite dimensional system originally obeys.



Usually, the proof of Nagy's unitary dilation theorem relies on a single example using the matrix method. The dilated matrix is constructed by the non-unitary operator $V$ and defect operator $D_V=\sqrt{I-V^\dagger V}$. It can be seen as a minimal block-encoding \cite{gilyen2019quantum,low2019hamiltonian} of any finite-dimensional matrix using only one ancilla qubit. However, the computational cost of implementing this block-encoding remains to be discussed. One can also construct block-encodings for specific matrices by employing additional ancilla qubits, which has proven to be a valuable tool in quantum computing \cite{gilyen2019quantum,low2019hamiltonian}.

\textit{Schr\"odingerisation} is a simple and generic  procedure designed to convert any linear dynamical system, encompassing both linear ordinary and partial differential equations, into Schr\"odinger-type partial differential equations in a higher space dimension \cite{jin2022quantum,jin2022quantumdetail}. In our study, we establish a direct connection between the Schr\"odingerisation approach and Nagy's unitary dilation theorem by providing a novel and alternative proof of this dilation theorem. We focus our attention on dilating operators of the form $V(t)=\exp(-At)$, which holds significant importance as it applies to solutions of numerous linear evolution equations. To ensure stability, we assume that the real parts of eigenvalues of $A$ are positive. In many cases, $A$ lacks an anti-Hermitian component, which results in $V(t)$ being non-unitary. This connection sheds new light on the understanding and applicability of the dilation theorem in such scenarios.

We demonstrate that the Schr\"odingerisation method can be viewed as a new constructive example of the dilation theorem by defining the dilated space. Prior to this, well-known and easily understandable matrix constructions, as illustrated in Section \ref{sec:bg:proof}, primarily dealt with discrete time instances $t=0,1,2,\dots$. For continuous time $t\in\mathbb{R}$, the dilation theorem had a somewhat theoretical proof that lacked intuition\cite{nagy2010harmonic}. In our proof, it becomes evident how the Schr\"odingerisation method dilates the operator $V(t)$, resulting in a unitary operator $U(t)$ that still forms a continuous semigroup concerning time $t$. Furthermore, this transformation solely relies on quantum Fourier transforms and inverse Fourier transforms, eliminating the need for complex mathematical operations like root-finding.

The Schr\"odingerisation formalism is versatile, naturally applying to continuous-variable quantum modes, enabling analog quantum computation \cite{jin2023analog}. It can also be adapted to qubits by discretising continuous variables. In the infinite-dimensional case, Schr\"odingerisation serves as an exact implementation of the dilation theorem, requiring only one additional quantum mode, representing the Fourier mode $\eta$. Following discretisation, the Schr\"odingerisation method can approximate the dilated unitary within precision $\delta$ using $\mathcal{O}(\log(1/\delta))$ ancilla qubits.

In comparing the Schr\"odingerisation method with block-encoding, we assess the complexities of simulating $V(T), T>0$ for qubit systems. To clarify, we apply the first-order Lie-Trotter decomposition. Notably, it appears that the complexity of the Schr\"odingerisation method increases at a faster rate with respect to the precision $\delta$ due to the presence of the large Fourier mode $\eta$. However, the complexity of block-encoding relies on $\lambda_{0}(A)$, the smallest real part of eigenvalues, which varies with the mesh size and consequently, $\delta$. We conducted specific calculations for the heat equation, in which case both Schr\"odingerisation and block-encoding exhibit comparable complexities.


The paper is organized as follows. In Section \ref{sec:bg}, we present the classical dilation theorem and give an existence proof. Moving on to Section \ref{sec:block}, the definition and some applications of block-encoding are given. In addition, we present a theorem on the block-encoding of the operator $V=e^{-At}$, as an example of the finite-dimensional case. Section \ref{sec:schrodingerisation} provides an overview of the Schr\"odingerisation approach, applicable to both infinite and finite cases for general PDEs. In this section, we also present a complexity theorem in the finite case. Finally, we conclude in Section \ref{sec:discussion} with a discussion of the results.

\textbf{Notations}: Throughout the paper, we use $\tilde{\mathcal{O}}$ to denote $\mathcal{O}$ where logarithmic terms are ignored and denote $f=\Omega(g)$ if $g=\mathcal{O}(f)$.

\section{Background}
\label{sec:bg}
In this section, we revisit the original Nagy's dilation theorem and show a concrete example which serves as a constructive proof of the theorem. The primary mathematical tool in quantum mechanics and dilation theory is the theory of Hilbert spaces. The reader can refer to \cite{griffiths2018introduction} for a brief review of quantum mechanics. To help the readers, we list the definitions of several notations used in this paper.
\begin{itemize}
    \item $\mathcal{H}$ represents a Hilbert space with inner product $\langle\cdot,\cdot\rangle$ consistent with $L^2$ norm. In the quantum mechanics context,  $|\psi\rangle \in \mathcal{H}$ represents a vector in $\mathcal{H}$ (a column vector), $\langle\psi|$ represents a vector of the dual Hilbert space of $\mathcal{H}$ (a row vector), we also use the bra-ket notation $\langle\psi|\phi\rangle \in \mathbb{C}$
 for the inner product of vectors $|\psi\rangle$ and $|\phi\rangle$.
    
    \item $B(\mathcal{H})$ represents the space of bounded operators acting on $\mathcal{H}$, $A^\dagger \in B(\mathcal{H})$ is the Hermitian conjugate of the operator $A \in B(\mathcal{H})$ and $I$ is the identity operator.

    \item $P_{\mathcal{H}} U \in B(\mathcal{H})$ represents the projection of $U \in B(\mathcal{H}_1)$, satisfying $\mathcal{H} \subset \mathcal{H}_1$ and $\langle P_{\mathcal{H}} U h, h^\prime \rangle = \langle U h, h^\prime \rangle$ for all $h, h^\prime \in \mathcal{H}$.
    
    \item $H \in B(\mathcal{H})$ is a Hermitian operator iff $H=H^\dagger$ and $U \in B(\mathcal{H})$ is a unitary operator iff $UU^\dagger = U^\dagger U = I$.
    
    \item $\left\| A \right\|_2$ and $\left\| A \right\|_{\max}$ represent the 2-norm and max-norm of $A\in B(\mathcal{H})$ respectively. $A$ is called a \textit{contraction} if $\left\| A \right\|_2 \leq 1$. $\lambda_0(A)$ denotes the smallest real part of all the eigenvalues of $A\in B(\mathcal{H})$. If $A$ is Hermitian, $\lambda_0(A) = \lambda_{\min}(A)$ is the smallest eigenvalue of $A$.
    
    \item \textit{Hamiltonian} denotes a Hermitian operator. Given a Hamiltonian $H$, the time evolution operator $U(t) = \exp(-iHt)$ is a unitary operator.
    
    \item \textit{qumodes} denotes continuous-variable quantum modes.
\end{itemize}

\subsection{Dilation theorem}
\label{sec:bg:dilation}

\begin{mythm}[Sz.-Nagy’s unitary dilation theorem \cite{nagy2010harmonic}]\label{thm:intro:nagy}
Let $V$ be a contraction on a Hilbert space $\mathcal{H}$, then there exists a Hilbert space $\mathcal{H}_1$ containing $\mathcal{H}$ and a unitary $U$ on $\mathcal{H}_1$, such that 
\begin{equation}
    V^k = P_{\mathcal{H}} U^k, \quad \text{for all } k = 0,1,2,...
\end{equation}
Moreover, when $\mathcal{H}_1$ is chosen as the smallest reducing subspace for $U$ that contains $\mathcal{H}$,
\begin{equation}\label{eqn:intro:nagy:min}
    \mathcal{H}_{1} = \bigvee_{n\in \mathbb{Z}} U^n\mathcal{H} := \overline{\myspan \{U^n\mathcal{H}, n\in\mathbb{Z}\}},
\end{equation}
$(\mathcal{H}_1, U)$ can be identified as a minimal dilation. These conditions determine $U$ up to an isomorphism.
\end{mythm}
In this situation we consider the dilation of a single contraction. It can also be developed in the context of operator semigroups. For instance, if $\{V(t)\}_{t\geq0} \subset B(\mathcal{H})$ is a continuous semigroup, i.e., $V(0)=I$, $V(s+t) = V(s)V(t)$, then there exists a continuous group $\{U(t)\}_{-\infty}^{\infty}$ of unitary operators on $\mathcal{H}_1 \supset \mathcal{H}$ such that
\begin{equation}\label{eqn:bg:dilation:semigroup}
    V(t) = P_{\mathcal{H}} U(t), \quad \text{for all } t \geq 0.
\end{equation}
We introduce the general Nagy's theorem and its constructive proof in the Appendix \ref{sec:appendix:nagy}. For more details, we refer the readers to \cite{nagy2010harmonic}.

Note that Nagy's theorem does not specify whether or not the dimension of $\mathcal{H}$ is finite or infinite, so it should be applicable to both cases. We will look at the finite and infinite cases separately and show how they are represented differently by quantum systems. Namely, in the finite dimensional case, it is dilation of Hilbert space constructed from qubit systems. While in the infinite dimensional case with an uncountable domain, the corresponding infinite-dimensional Hilbert space is constructed from continuous-variable quantum modes, or `qumodes' \cite{braunstein2005quantum, liu2016power}.

Given $\mathcal{H}$ the $n$-dimensional Hilbert space with an orthonormal basis, $h\in \mathcal{H}$ can be represented by an $n$-dimensional vector and the operator $V$ acting on $\mathcal{H}$ can be regarded as an $n\times n$ matrix. Block-encoding is a powerful technique in this case to represent arbitrary matrices using unitary matrices. However, when $\mathcal{H}$ is infinite-dimensional, the operator $V$ {\it cannot} be represented by a finite dimensional matrix. If the operator $\mathcal{H}$ has a spectrum in the uncountable domain, then it cannot be represented by a matrix at all. 

The proof presented here is a constructive proof, which is done by giving a specific example of the dilation.

\subsection{Proof of the dilation theorem}
\label{sec:bg:proof}

Given $V \in \mathcal{H}$ a contraction, that is, $\left\| V \right\| \leq 1$, one can define $D_V = \sqrt{I-V^\dagger V}$ since $I-V^\dagger V \geq 0$. The simple construction
\begin{equation}\label{eqn:nagy:unitarymatrix}
    U := \begin{bmatrix}
        V & D_{V^\dagger} \\ D_{V} & -V^\dagger
    \end{bmatrix}
\end{equation}
gives rise to a unitary operator on $\mathcal{H}_1 = \mathcal{H}\otimes\mathcal{H}$. $\mathcal{H}$ is embedded in $\mathcal{H}_1$ trivially by $h_0 \mapsto \begin{bmatrix} h_0\\ 0 \end{bmatrix}$, and the action of $U$ on it is defined as
\begin{equation}
    U \begin{bmatrix} h_0\\ 0 \end{bmatrix} =  \begin{bmatrix} V h_0\\ D_V h_0 \end{bmatrix}.
\end{equation}
It is easy to check that for $h_0, h_0^\prime$, one has $\langle U h_0, h_0^\prime \rangle = \langle Vh_0, h_0^\prime \rangle$, meaning that
\begin{equation}
    V = P_{\mathcal{H}} U.
\end{equation}
Here $P_{\mathcal{H}}$ can be seen as the orthogonal projection of $\mathcal{H}\otimes\mathcal{H}$ onto $\mathcal{H}\otimes\{0\} \simeq \mathcal{H}$. $U$ can be seen as a $(1,1)-$block-encoding of $V$, which will be discussed in Section \ref{sec:block}.

Similarly, one can define $\mathcal{H}_1 = \mathcal{H}^{N+1} = \mathcal{H}\otimes\cdots\otimes\mathcal{H}$ and consider the following matrix
\begin{equation}
    U := \begin{bmatrix}
        V      & 0      & 0      & \cdots & 0 & D_{V^\dagger} \\ 
        D_{V}  & 0      & 0      & \cdots & 0 & -V^\dagger    \\
        0      & I      & 0      & \cdots & 0 & 0             \\
        0      & 0      & I      &        &   & 0             \\
        \vdots & \vdots &        & \ddots &   & \vdots        \\
        0      & 0      & \cdots & 0      & I & 0
    \end{bmatrix}.
\end{equation}
It is easy to check that $U$ is unitary, thus 
\begin{equation}
    U^k = \begin{bmatrix}
        V^k & * \\ * & *
    \end{bmatrix} \quad 
    \Rightarrow \quad V^k = P_{\mathcal{H}} U^k, \quad k = 0,\dots,N.
\end{equation}
This idea can be pushed further by taking $\mathcal{H}_1 = \otimes_{n\in\mathbb{Z}}\mathcal{H}$ and $U$ to be an infinite operator matrix. The minimal dilation is then obtained by restricting $U$ to the reducing subspace $\bigvee_{n\in \mathbb{Z}} U^n\mathcal{H}$. One can use this example to prove the case $k=0,1,2,\dots$, but for continuous time case \eqref{eqn:bg:dilation:semigroup}, one can give a more abstract proof, as shown in Appendix \ref{sec:appendix:nagy}.

\section{Block-encoding}
\label{sec:block}
Block-encoding is a general method to ``unitarize'' a non-unitary matrix $A\in\mathbb{C}^{n\times n}$ ($n=2^m$). Since it can be difficult to implement a unitary matrix $U_A$ to block encode $A$ exactly, 
it is sufficient if one can find $U_A$ to block encode $A$ up to some error $\delta$. Block-encoding is defined in Definition  \ref{def:block} \cite{chakraborty2018power}.

\begin{mydef}\label{def:block}
Given the $m$-qubit matrix $A$ with $n=2^m$ and a threshold $\delta \geq 0 $, if one can find $\alpha>0$ and an $(l+m)$-qubit unitary matrix $U_{A}$ such that
\begin{equation}
    \left\| A - \alpha (\langle 0^l| \otimes I) U_{A} (|0^l\rangle \otimes I) \right\|_2 \leq \delta,
\end{equation}
then $U_{A}$ is called an $(\alpha, l, \delta)$-block-encoding of $A$. In particular, when the block-encoding is exact with $\delta = 0$, $U_{A}$ is called an $(\alpha,l)$-block-encoding of $A$. 
\end{mydef}

The construction in \eqref{eqn:nagy:unitarymatrix} in Section \ref{sec:bg:proof} guarantees a minimal block-encoding of any arbitrary $n$-dimensional contraction $V$ using only {\it one} ancilla qubit.  From the perspective of constructing quantum circuits, given an initial state $|\psi\rangle \in \mathcal{H}=\mathbb{C}^{2^m}$, one ancilla qubit is used to produce $|0\rangle \otimes |\psi\rangle \in \mathcal{H}_1=\mathbb{C}^{2^{m+1}}$. Afterwards, the operator $U_{A}$ is applied. The ancilla qubit is to be measured at the end of the circuit and one obtains the desired state $A|\psi\rangle$, which corresponds to applying the projection operator $P_\mathcal{H}$ defined as
\begin{equation}
    P_{\mathcal{H}}U_{A} \cdot |\psi\rangle := (\langle 0| \otimes I)U_A(|0\rangle \otimes I) |\psi\rangle.
\end{equation}
However, implementing the off-diagonal blocks defect operators $\sqrt{I-A^\dagger A}$, $\sqrt{I-A A^\dagger}$ is non-trivial. 

If $A$ has some special structure, for example, it is a sparse Hermitian matrix, denoted by $H$, one can develop some easier-to-implement block-encoding strategies with more ancilla qubits. For example, given the sparse access oracle of a sparse Hermitian (also called a Hamiltonian) matrix $H$, one can construct a block-encoding $U_H$ using the following lemma.
\begin{mylemma}[{\cite[Lemma 6]{low2019hamiltonian}}]\label{lemma:block:hamiltonian}
    Let the oracle $O_{H}$ specify an $s$-sparse Hamiltonian $H \in\mathbb{C}^{2^m\times 2^m}$ with max-norm $\|H\|_{\max}$ and $O_{F}$ specify the column index of its non-zero elements,
    \begin{equation}
        O_{H} |j\rangle |k\rangle |z\rangle = |j\rangle |k\rangle |z\oplus H_{jk}\rangle, \quad O_{F} |j\rangle |l\rangle = |j\rangle |f(j,l)\rangle.
    \end{equation}
    Then the oracles encoding $(\langle 0^{m+2} | \otimes I)U_H(|0^{m+2} \rangle \otimes I ) = \frac{H}{s\|H\|_{\max}}$ can be implemented using $\mathcal{O}(1)$ queries.
\end{mylemma}
In addition, one can construct block-encodings to perform matrix operations, including summation, multiplication, inverse, and polynomial transformation \cite{camps2022explicit, gilyen2019quantum}. Several properties that will be utilized in our work are introduced in Appendix \ref{sec:lemmas}.

We turn our attention to the operator of the form $V(t)=\exp(-At)$, which can be seen as the solution operator of the general evolution system
\begin{equation}\label{mod:origin}
    u_t = -Au, \quad u(t=0) = u_0, \quad t\in [0,T],
\end{equation}
and from now on, we require $\lambda_0(A) > 0$, namely, the real parts of all the eigenvalues of $A$ are positive. While considering a system of ODEs with $n$ variables, $\mathcal{H}$ is finite dimensional and $A\in\mathbb{C}^{n\times n}$. Given a Hermitian matrix $A$, generally one cannot use Lemma \ref{lemma:block:hamiltonian} to encode $V(T)$ at the final time $T$ without the sparse access oracle. However, one can construct a block-encoding $U(T)$ of the solution operator $V(T)$ by the quantum singular value transformation (QSVT) as shown in Appendix \ref{sec:lemmas}, Lemma \ref{lemma:block:expAt}. For non-Hermitian $A$, if $A=W\Sigma V^\dagger$ is a singular value decomposition, QSVT gives $P^{(SV)}(A):=WP(\Sigma)V^\dagger$ instead of the polynomial transformation $P(A)$. Therefore one needs to split $A$ into a linear combination of Hermitian operators $H_1,H_2$ with $A=H_1+iH_2$,
\begin{equation}
    H_1 = \frac{A+A^\dagger}{2}, \quad H_2 = \frac{A-A^\dagger}{2i}, \quad H_1^\dagger = H_1, \quad H_2^\dagger = H_2, \quad \lambda_{\min}(H_1) = \lambda_0(A) > 0.
\end{equation}
Making a Trotter decomposition, $V(T)=e^{-AT}$ can be simulated step by step through $e^{-H_1\Delta t}$ with Appendix \ref{sec:lemmas}, Lemma \ref{lemma:block:expAt} and $e^{-iH_2 \Delta t}$ with Appendix \ref{sec:lemmas}, Lemma \ref{lemma:block:expiHt} separately. The results are summarized in Theorem \ref{thm:block:expAt:nonHer}. For the details of proof, see Appendix \ref{sec:proof:block:expAt:nonHer}.

\begin{mythm}\label{thm:block:expAt:nonHer}
Consider solving Equation \eqref{mod:origin} where $A\in\mathbb{C}^{2^m\times 2^m}$ is a non-Hermitian matrix with the smallest real part of eigenvalues $\lambda_{0}(A)>0$. Suppose that $A=H_1+iH_2$ with $H_1$, $H_2$ Hermitian matrices and one has sparse access to oracles $O_{H_1}$, $O_{H_2}$, $O_{u}$. For any $T > 0$ and $\delta < 1/2$, there exists a quantum algorithm that outputs an $\delta$-approximation of $|u(T)\rangle$ with $\Omega(1)$ success probability, using
\begin{equation}\label{eqn:block:expAt:Q}
    \tilde{\mathcal{O}} \left( \left(\frac{\left\|u(0)\right\|}{\left\|u(T)\right\|}\right)^{2} \frac{\tau^2}{\delta} \left( 1 + \frac{s(A)\|A\|_{\max}}{\lambda_{0}(A)} \right) \right)
\end{equation}
queries to $O_{H_1}$, $O_{H_2}$, their inverse and controlled versions, $\mathcal{O}\left( \frac{\left\|u(0)\right\|}{\left\|u(T)\right\|}\right)$ queries to the controlled versions of the initial state oracle $O_u$, $\mathcal{O}(m)$ ancilla qubits, and
\begin{equation}
    \tilde{\mathcal{O}} \left( \left(\frac{\left\|u(0)\right\|}{\left\|u(T)\right\|}\right)^{2} \frac{\tau^2}{\delta} \left( m + \frac{s(A)\|A\|_{\max}}{\lambda_{0}(A)} \right) \right)
\end{equation}
one- and two- qubit gates with $\tau=s(A)\|A\|_{\max}T$.
\end{mythm}

\begin{myremark}\label{remark:block:expAt:nonHer}
 The parameter $\left\|u(0)\right\|/\left\|u(T)\right\|$ comes from oblivious amplitude amplification to approximate $|u(T)\rangle$ with $\Omega(1)$ success probability. The parameter $\tau^2/\delta$ comes from the number of Trotter decomposition segments 
 $K$. For each segment, we apply Lemma \ref{lemma:block:expAt} with $T^\prime=T/K$ and $\tilde{T}^\prime=\log(1/(\delta/K))$. The parameter $1/\lambda_{0}(A)$ comes from constructing a block-encoding of $(I-H_1)$, which is utilized to construct the block-encoding of $e^{-H_1T^\prime}$. If $H_1$ satisfies other better conditions, for example, the square root or the eigenvalues and eigenstates of $H_1$ are known, the complexity can be independent of $\lambda_{0}(A)$. For more details, we refer the readers to \cite{an2022theory}.
\end{myremark}

\section{Schr\"odingerisation}
\label{sec:schrodingerisation}
We first make a remark that the Sz.-Nagy dilation theorem guarantees that an $n$-dimensional non-unitary matrix can be dilated to a unitary matrix of twice the dimension, which only requires an augmentation of the system by one ancilla qubit. This method is also referred to as `qubitisation' \cite{low2019hamiltonian}. Now it is intriguing to ask if there is an analogue procedure for the case of infinite-dimensional systems. In this case, we see that instead of adding a single qubit, we add a single infinite-dimensional quantum system, a \textit{qumode}. This can be considered as a \textit{qumodisation} procedure. This is a continuous-variable version of the Schr\"odingersation procedure.

Considering operators of the form $V(t)=\exp(-At)$, the solution operator of the general evolution system \eqref{mod:origin}. It has been shown in Section \ref{sec:block} that if $\mathcal{H}$ is finite dimensional, a block-encoding of $V(T)$ can be constructed. Nevertheless, when $A$ is a linear differential operator acting on a function, which often appears in PDEs, the dimension of $\mathcal{H}$ is {\it infinite}. An alternative method called \textit{Schr\"odingerisation} was introduced in \cite{jin2022quantum,jin2022quantumdetail} to solve equation \eqref{mod:origin}. The Schr\"odingerisation method involves transforming the original equation into a Schr\"odinger-type equation in one higher dimension using a simple yet powerful technique called the \textit{warped phase transformation}, which transform $u(t)$ to $e^{-p}u(t)$ for $p>0$.  $V(t)$ is transformed to the unitary solution operator of the deduced Schr\"odinger-type equation. The Schr\"odingerisation formalism is not only naturally applicable to continuous-variable quantum modes, but also applicable to qubits by discretising those continuous variables. For a description of Schr\"odingerisation entirely in the qumode language, see \cite{jin2023analog}.

In this section, we focus on solving equation \eqref{mod:origin} with the Schr\"odingerisation approach. We give the details of this method and establish its connection to Nagy's dilation theorem for both continuous-variable quantum modes (qumodes) and qubits.

\subsection{Qumodisation}
\label{sec:schro:cv}

While considering continuous-variable quantum modes, we denote $A_{CV}$ to be an infinite operator with $\lambda_0(A)>0$ and make the decomposition
\begin{equation}
    A_{CV} = H_1+iH_2, \quad H_1 = \frac{A_{CV}+A_{CV}^\dagger}{2}, \quad H_2 = \frac{A_{CV} - A_{CV}^\dagger}{2i}.
\end{equation}
Introduce a real one-dimensional variable $p>0$ and define
\begin{equation}
    w(t,p) = e^{-p} u(t),
\end{equation}
which is called the \textit{warped phase transformation}. The original solution $u(t)$ of equation \eqref{mod:origin} can be recovered using $u(t)=\int_{0}^{\infty} w(t,p) dp$ or $u(t)=e^{p}w(t,p)$ for any $p>0$. We extend the domain of $p$ to $(-\infty,+\infty)$ with $w = w(t,p)$ satisfying 
\begin{equation}
    \partial_t w = -(H_1 + iH_2) w = H_1 \partial_p w - iH_2 w,
\end{equation}
with evenly extended initial condition $w(0,p) = e^{-|p|}u_0$. Let $\tilde{w} = \tilde{w}(t,\eta)$ be the Fourier transform of $w$ in $p$ and $\eta \in \mathbb{R}$ be the Fourier mode. Then $\tilde{w}$ satisfies the following Schr\"odinger-type equation
\begin{equation}\label{mod:schro:cv}
    i\partial_t \tilde{w} = (\eta H_1 + H_2) \tilde{w}, \quad \tilde{w}(0,\eta) = \mathcal{F}w(0,\eta) = \frac{2}{\eta^2 + 1} u_0.
\end{equation}
The solution operator of the original model \eqref{mod:origin} is $V(t):=\exp(-A_{CV}t)$ for $t \geq 0$, which can be extended to
\begin{equation}
    V(t) = \exp(-H_1|t|-iH_2t), \quad t\in\mathbb{R},
\end{equation}
with $V(t):=V(-t)^\dagger$ for $t<0$. The solution operator of the Schr\"odinger-type equation \eqref{mod:schro:cv} is
\begin{equation}
    U_{CV}(t) = \exp(-i(\eta H_1 + H_2)t), \quad t \in \mathbb{R}.
\end{equation}
In the following text, we show that $U_{CV}(t)$ is actually a dilation of $V(t)$ as defined in \eqref{eqn:bg:dilation:semigroup}. 

\begin{mythm}[Dilation theorem via Schr\"odingerisation]
Let $V(t) = \exp(-At)$ with $\lambda_0(A)>0$, acting on $\mathcal{H}$. $\{V(t)\}_{t\geq 0}$ is a continuous semigroup of contractions. The unitary operators $\{U_{CV}(t)\}_{t\geq0}$ acting on a larger space $\mathcal{H}_1 \supset \mathcal{H}$
\begin{equation}
    U_{CV}(t) = \exp(-i(\eta H_1 + H_2)t)
\end{equation}
form a unitary dilation of $\{V(t)\}_{t\geq 0}$, meaning that
\begin{equation}
    V(t) = P_{\mathcal{H}} U(t), \quad \text{for all } t \geq 0.
\end{equation}
\end{mythm}

\begin{proof}
Let us define the space $\mathcal{H}_1$ of mappings $h(\eta)$ from $\mathbb{R}$ to $\mathcal{H}$ with the following bilinear form
\begin{equation}\label{eqn:schro:cv:bilinearH}
    \langle h, h^\prime \rangle := \frac{1}{2\pi} \int_{\mathbb{R}} f(\eta) \langle h(\eta), h^\prime(\eta) \rangle \dif \eta, \quad [h=h(\eta),h^\prime=h^\prime(\eta)], \quad f(\eta) = \frac{2}{\eta^2+1}.
\end{equation}
$h(\eta)\in\mathcal{H}_1$ if $\langle h, h\rangle$ is well defined. $\mathcal{H}$ is embedded in $\mathcal{H}_1$ trivially by $h_0 \mapsto h = h(\eta) \equiv h_0 $ and an operator $U$ acting on $\mathcal{H}_1$ is projected to the operator $P_{\mathcal{H}}U$ acting on $\mathcal{H}$ defined by
\begin{equation}
    P_{\mathcal{H}}U \cdot h_0 := \frac{1}{2\pi} \int_{\mathbb{R}} f(\eta) U(h(\eta)) d\eta.
\end{equation}
The definitions are well defined since for $h_0, h_0^\prime \in \mathcal{H}$,
\begin{equation}
    \langle h, h^\prime \rangle = \frac{1}{2\pi} \int_{\mathbb{R}} f(\eta) \langle h_0, h_0^\prime \rangle \dif \eta = \langle h_0, h_0^\prime \rangle .
\end{equation}
Clearly $U_{CV}(t)$ is a unitary operator. Given $h_0, h_0^\prime \in \mathcal{H}$, we point out that
\begin{equation}\label{eqn:schro:cv:projU}
    P_{\mathcal{H}} U_{CV}(t) \cdot h_0 = \frac{1}{2\pi} \int_{\mathbb{R}} \frac{2}{\eta^2+1}  e^{-i(\eta H_1+H_2)t} h_0 d\eta = e^{- H_1|t| - iH_2t} h_0,
\end{equation}
which is proved in Appendix \ref{sec:proof:residue}. Therefore $U_{CV}(t)$ is exactly a dilation of $V(t)$ since $P_{\mathcal{H}}U_{CV}(t)=V(t)$. From the perspective of solving equations, we apply the inverse Fourier transformation to recover $w(t,p)=\mathcal{F}^{-1}\tilde{w}(t,\eta)$ after obtaining $\tilde{w}(t,\eta) = U_{CV}(t)\tilde{w}(0,\eta)$.
\end{proof}


\begin{myremark}
Schr\"odingerisation is an explicit construction of Nagy's dilation theorem. In comparison to certain existing methods, such as the approach presented in \cite{shamovich2017dilations}, where the contraction is incorporated into an input/state/output (i/s/o) linear system, Schr\"odingerisation elevates the system by only one dimension and offers a straightforward implementation. Furthermore, Schr\"odingerisation can be seamlessly extended to encompass multi-parameter semigroup of commuting operators.
\end{myremark}
 
\begin{myexample}[Heat equation]
We take the heat equation as an example.
\begin{equation}\label{mod:schro:cv:heat}
    \partial_t u(t,x) = - A_{CV} u(t,x) = - (-\nabla_x^2 + V(x)) u(t,x), \quad x\in\mathbb{R}^d.
\end{equation}
Using Schr\"odingerisation, the heat equation is transformed into a Schr\"odinger equation by the warped phase transformation $w(t,x,p) = e^{-p} u(t,x)$ and the Fourier transform $\tilde{w}=\mathcal{F}w$:
\begin{equation}
    \partial_t w = (-\nabla_x^2 + V(x)) \partial_p w, \quad i\partial_t \tilde{w} = \eta(-\nabla_x^2 + V(x)) \tilde{w},
\end{equation}
where $\eta\in\mathbb{R}$ is the Fourier mode. {\it The Schr\"odingerisation formalism is naturally applicable to qumodes}. These are the analogue or continuous counterparts to qubits. Here the solution $u(t,x)$ are embedded in continuous-variable quantum states $ |u(t)\rangle \propto \int u(t,x)|x\rangle dx$ where $\{|x\rangle\}$ is an orthonormal basis set in infinite dimensional Hilbert space since $x$ spans $\mathbb{R}^d$ where $d$ is the number of spatial dimensions in $x$. If $x$ represents the position of the wavefunction $u(x,t)$, the corresponding operator is $\hat{x}$ where $\hat{x}|x\rangle=x|x\rangle$. The conjugate operator is the momentum operator $\hat{p}$ where $[\hat{x}, \hat{p}]=i$ and we can represent $\hat{p}=-i \partial/ \partial x$. This means the operator $\exp(-At)$ can be encoded in the unitary
\begin{align}
    U_{CV}(t)=\exp(i(\hat{p}^2+V(\hat{x})) \otimes \hat{\eta}t).
\end{align}
Here $\hat{\eta}$ can be chosen to be any one-mode operator, e.g. $\hat{\eta}=\hat{x}$. In this case, {\it no discretisation of the system in any variable is necessary}. In principle, this is a much more accurate simulation of the PDE since it deals directly with the continuous nature of the PDEs, and does not depend on the details of the discretisation schemes.
\end{myexample}

In the above example, we showed dilation can be used in the context of continuous-variable quantum modes, or `qumodes'. This is most appropriate when we want to embed $\exp(-At)$ into a unitary operator, when $A$ itself can be infinite-dimensional. This is true when solving PDEs for example. Upon discretisation of the differential operators $A$, one has a  representation by finite-dimensional matrices. In the absence of discretisation schemes, we keep the continuous nature of the differential operators, so $A=H_1+iH_2$ acts on qumodes instead of qubits, where $H_1, H_2$ are Hermitian operators also acting on qumodes. Then $\exp(-At)$ can be embedded into the unitary operator
\begin{align}
    U_{CV}(t)=\exp(-i(H_1 \otimes \hat{\eta}+H_2 \otimes I)t).
\end{align}
Here we can ask a similar question as for the qubit case, namely what is the cost in simulating $U_{CV}(t)$? Here in the continuous-variable case, we don't have a direct analogue of sparse-access or block-access to an infinite-dimensional operator $A$ (which has infinite norm), where we can directly apply the digital quantum simulation algorithms. In this case, it is more appropriate to consider {\it analogue quantum simulation}, where instead of considering primitive gates, we consider gates that naturally realise the appropriate unitary operation. For instance, if one has access to $U_1(t)=\exp(-iH_1 \otimes \hat{\eta}t)$ and $U_2(t)=\exp(-i H_2 \otimes It)$ and $[H_1, H_2]=0$, then one only needs to implement $U_1(t)U_2(t)$. This would be the case for instance when $H_1$ and $H_2$ are both operators only in $\hat{x}$ (true for ODE problems) or only in $\hat{p}$ (true for homogeneous PDE problems with only constant-valued coefficients) . More complicated scenarios can be considered on a case by case basis.\\


\subsection{Schr\"odingerisation for qubit systems}
\label{sec:schro:dv}

The Schr\"odingerisation formalism is not only applicable to qumodes, but also to qubits. We first discretise the system in $x$, then the origin PDE is turned  to a system of ODEs with $n$ variables and the matrix $A_{DV}\in\mathbb{C}^{n\times n}$
\begin{equation}\label{mod:origin:dv}
    \frac{d\mathbf{u}}{dt} = -A_{DV}\mathbf{u}, \quad A_{DV} = \begin{bmatrix}
    a_{11} & \cdots & a_{1n} \\ \vdots & \ddots & \vdots \\ a_{n1} & \cdots & a_{nn}
    \end{bmatrix}, \quad \mathbf{u} = \begin{bmatrix}
    u_1 \\ \vdots \\ u_n
    \end{bmatrix}, \quad \mathbf{u}(0)=\mathbf{u}_0 \in \mathbb{C}^{n}.
\end{equation}
Similarly as in Section \ref{sec:schro:cv}, one applies the warped phase transformation $\mathbf{w}(t,p) = e^{-p}\mathbf{u}(t)$ and Fourier transform $\tilde{\mathbf{w}}(t,\eta) = \mathcal{F}\mathbf{w}(t,p)$, then $\mathbf{w}$, $\tilde{\mathbf{w}}$ satisfy the following equations
\begin{equation}
    \partial_t \mathbf{w} = -A_{DV}  \mathbf{w} = H_1 \partial_p \mathbf{w} - i H_2 \mathbf{w}, \quad i \partial_t \tilde{\mathbf{w}} = (\eta H_1 + H_2) \tilde{\mathbf{w}},
\end{equation}
where $H_1 = (A_{DV} + A_{DV}^\dagger)/2$, $H_2 = (A_{DV} - A_{DV}^\dagger)/(2i)$. To solve these equations with qubits, one proceeds by discretising $\eta$ with a mesh size $\Delta \eta = 2L/N$ in the domain $[-L,L]$, with $N$ a positive integer and $L>0$. Usually one assumes the computational domain to be of $\mathcal{O}(1)$, so we define $\eta^\prime = \eta / L \in [-1,1]$, $\Delta \eta^\prime = 2/N$ and $D=L\diag(\eta_1^\prime,\dots,\eta_N^\prime)$ with entries $\eta_j^\prime= \eta_j/L = -1+j\Delta \eta^\prime$, to obtain
\begin{equation}\label{mod:schro:dv}
    i\frac{d}{dt} \tilde{\mathbf{w}}_{DV} = (H_1\otimes D + H_2\otimes I) \tilde{\mathbf{w}}_{DV} = H_{total, DV} \tilde{\mathbf{w}}_{DV},
\end{equation}
where $\tilde{\mathbf{w}}_{DV}(t):=[\tilde{\mathbf{w}}(t,\eta_1),\dots,\tilde{\mathbf{w}}(t,\eta_N)]$. 
This is most appropriate if one wants to embed $\tilde{\mathbf{w}}$ into qubits. The solution operator of the origin model \eqref{mod:origin:dv} is $V(t) = \exp(-A_{DV}t)$, while the solution operator of the Schr\"odinger-type equation \eqref{mod:schro:dv} is
\begin{equation}
    U_{DV}(t) = \exp(-i(H_1\otimes D + H_2\otimes I) t).
\end{equation}
In the following text, we show that $U_{DV} (t)$ is an approximated dilation of $V(t)$.

\begin{mythm}
Let $V(t) = \exp(-At)$ with $\lambda_0(A) > 0$, acting on $\mathcal{H}$. $\{V(t)\}_{t\geq 0}$ is a continuous semigroup of contractions. The unitary operators $\{U_{DV}(t)\}_{t\geq0}$ acting on a larger space $\mathcal{H}_1 \supset \mathcal{H}$
\begin{equation}
    U_{DV}(t) = \exp(-i(H_1\otimes D + H_2\otimes I) t)
\end{equation}
is close to a unitary dilation of $\{V(t)\}_{t\geq 0}$ within error $\mathcal{O}(\delta)$, meaning that
\begin{equation}
    \left\| P_{\mathcal{H}} U(t) - V(t) \right\| = \mathcal{O}(\delta).
\end{equation}
\end{mythm}

\begin{proof}
Let us define the space $\mathcal{H}_1$ to be a subspace of $\mathcal{H}\times\mathbb{C}^{N}$ with the following bilinear form
\begin{equation}\label{eqn:schro:dv:bilinearH}
    \langle h, h^\prime \rangle := \sum_{j=1}^{N} \frac{f(L\eta_j^\prime) L\Delta \eta^\prime}{2\pi} \langle h_j, h^\prime_j \rangle , \quad h=(h_1,\cdots,h_N) \in \mathcal{H}_1.
\end{equation}
$\mathcal{H}$ is embedded in $\mathcal{H}_1$ trivially by 
\begin{equation}
    h_0 \mapsto h = (h_1,\cdots,h_N) = ( h_0, \cdots, h_0).
\end{equation}
An operator $U$ acting on $\mathcal{H}_1$ is projected to the operator $P_{\mathcal{H}}U$ acting on $\mathcal{H}$ defined by
\begin{equation}
    P_{\mathcal{H}}U \cdot h_0 := \sum_{j=1}^{N} \frac{f(L\eta_j^\prime)L\Delta \eta^\prime}{2\pi} U(h)_j.
\end{equation}
The definitions are well defined since for $h_0, h_0^\prime \in \mathcal{H}$ one has
\begin{equation}
    \langle h, h^\prime \rangle = 
    \left(\sum_{j=1}^{N} \frac{f(L\eta_j^\prime)L\Delta \eta^\prime}{2\pi}\right) \langle h_0, h^\prime_0 \rangle, \quad \langle h, h^\prime \rangle = 0  \Leftrightarrow \langle h_0, h_0^\prime \rangle=0.
\end{equation}
Clearly $U_{DV}(t)$ is a unitary operator. Given $h_0, h_0^\prime \in \mathcal{H}$ and equation \eqref{eqn:schro:cv:projU} one gets
\begin{equation}\label{eqn:schro:dv:projU}
    \left\| P_{\mathcal{H}} U_{DV}(t) h_0 - V(t) h_0 \right\| = \left\| \sum_{j=1}^{N} \frac{f(L\eta_j^\prime)L\Delta \eta^\prime}{2\pi} e^{-i(L\eta_j^\prime H_1+H_2)t}h_0 - \int_{\mathbb{R}} \frac{f(\eta)}{2\pi} e^{-i(\eta H_1+H_2)t} h_0 d\eta \right\| = \mathcal{O}(\delta),
\end{equation}
if one chooses $L=\mathcal{O}(\frac{1}{\delta})$, 
$L\Delta \eta^\prime = \mathcal{O}(\delta)$. Therefore $U_{DV}(t)$ is a dilation of $V(t)$ up to precision $\delta$.

\end{proof}

Given sparse access of $H_1$ and $H_2$, one can apply the Schr\"odingerisation method to construct a block-encoding of $V(t) = \exp(-A_{DV}T)$.  Instead of simulating $e^{-H_1T}$ with Lemma \ref{lemma:block:expAt} as in Section \ref{sec:block}, we use the linear combination of unitaries (LCU) approach as shown in Lemma \ref{lemma:block:lcu} to approximate
\begin{equation}
    P_{\mathcal{H}}U_{DV}(T)= \sum_{j=1}^{N} \frac{f(\eta_j)\Delta \eta}{2\pi} e^{-i(\eta_j H_1+H_2)T} \triangleq \sum_{j=1}^{N} y_{j} U_j.
\end{equation}
To implement each $U_j$, we apply the Lie-Trotter product formula with $K$ segments. The results are summarized in Theorem \ref{thm:schro:dv:expAt}. For the details of proof, see Appendix \ref{sec:proof:schro:expAt}.

\begin{mythm}\label{thm:schro:dv:expAt}
    Given sparse access to the $2^m \times 2^m$ matrices $H_{1}$, $H_{2}$, the LCU coefficient oracle $O_{coef}$ and the unitary $U_{initial}$ that prepares the initial quantum state $|u(0)\rangle$ to precision $\delta$. With the Schr\"odingerisation approach, the state $|u(T)\rangle$ can be prepared with query complexity
    \begin{equation}
        N_{Query,Schr} = \tilde{\mathcal{O}} \left(  \left(\frac{\left\|u(0)\right\|}{\left\|u(T)\right\|}\right)^{4} \frac{\tau^2}{\delta^{3}} \right), 
    \end{equation}
    $\mathcal{O}(m+\log(\tau/\delta))$ ancilla qubits and
    \begin{equation}
        N_{Gates,Schr} = \tilde{\mathcal{O}}\left( m  \left(\frac{\left\|u(0)\right\|}{\left\|u(T)\right\|}\right)^{4} \frac{\tau^2}{\delta^{3}} \right)
    \end{equation}
    additional two-qubit gates with $\tau = s(A)\|A\|_{\max}T$.
\end{mythm}

\begin{myremark}
We use the first order Lie-Trotter product formula to give an explicit result. Higher order product formulas can also be applied and the results depend on the order as shown in \cite{an2023linear}. 
\end{myremark}

Compared to Theorem \ref{thm:block:expAt:nonHer}, the complexity of the Schr\"odingerisation method depends on $\delta^{-3}$ instead of $\delta^{-1}$. This comes from the truncation of Fourier modes which leads to a larger number of segments in the Trotter decomposition. On the other hand, the Schr\"odingerisation method is independent of $\lambda_{0}(A)$. In the case of solving a PDE with finite difference methods, $\lambda_{0}(A)$ depends on the mesh size, and thereby depends on $\delta$. 

\begin{myexample}[Heat equation]
If we consider the heat equation \eqref{mod:schro:cv:heat}, then
\begin{equation}\label{eqn:schro:dv:heat:eigen}
    \lambda_{0}(A_{DV}) = \mathcal{O}\left( \frac{1}{h^2} \right) = \mathcal{O}\left( \frac{1}{(\delta^\prime)^2} \right), \quad s(A_{DV}) = 3, \quad \|A_{DV}\|_{\max} = \mathcal{O}(1),
\end{equation}
where $\delta^\prime = \delta \left\|u(T)\right\|/\left\|u(0)\right\|$. Substituting \eqref{eqn:schro:dv:heat:eigen} into \eqref{eqn:block:expAt:Q}, we find that Theorem \ref{thm:block:expAt:nonHer} and Theorem \ref{thm:schro:dv:expAt} give comparable complexities.

\end{myexample}

The Schr\"odingerisation formalism can also be used in hybrid continuous-variable discrete-variable settings, where one can have part of the system represented by qubits and the other part by qumodes. These can arise when one chooses to only discretise $A$ but not $\eta$, in which case  $U_{hyb}(t)=\exp(-i A_{DV} \otimes \hat{\eta}t)$. If one chooses only to discretise $\eta$ but not $A$, then  $U_{hyb}(t)=\exp(-i(\hat{p}^2+V(\hat{x})) \otimes Dt)$.

\section{Discussion}
\label{sec:discussion}
The basis of quantum computation lies in quantum simulation, where the system obeys quantum dynamics evolving under Schr\"odinger's equation. Without interaction with the environment, this restricts the dynamics to evolve under unitary operations. Nagy's Theorem provides a pivotal link between contraction operators and their unitary dilation. This theorem has far-reaching implications, applicable to both finite and infinite-dimensional cases. In our exploration, we study these distinct scenarios, demonstrating distinct representations by quantum systems. Specifically, the finite-dimensional case involves dilation of Hilbert spaces constructed from qubit systems, while the infinite-dimensional case employs continuous-variable quantum modes, or 'qumodes'. Moreover, we also establish Schr\"odingerization as an alternative constructive proof of the dilation theorem.

When $\mathcal{H}$ is finite-dimensional, block-encoding emerges as a powerful technique for representing arbitrary matrices using unitary matrices. However, the transition to infinite-dimensional $\mathcal{H}$ presents a challenge, as $V(t)$ can no longer be encapsulated within a finite-dimensional matrix representation. Herein, the Schr\"odingerisation method proves invaluable. Specifically, for operators in the form of $V(t) = \exp(-At)$, which govern evolution systems analogous to \eqref{mod:origin}, the Schr\"odingerisation approach provides a transformative solution. This method effectively elevates the equation into a Schr\"odinger-type equation in only one higher dimension. Consequently, the unitary solution operator for the derived Schr\"odinger-type equation encompasses the transformed operator $V(t)$. This technique is not only naturally compatible with continuous-variable quantum modes but can also be extended to qubits through suitable discretisation.

In the context of the continuous-variable version, the Schr\"odingerisation procedure necessitates the utilization of a single qumode, analogous to the qubitisation process, and we called this qumodisation. This qumode is introduced through a warped phase transformation coupled with the Fourier transform. In contrast, when considering the qubit version, the Schr\"odingerisation approach mandates the incorporation of approximately $\mathcal{O}(\log(1/\delta))$ ancilla qubits due to the discretisation of the continuous variable. However, this version is characterized by a relatively simpler implementation.

Comparing this with the block-encoding technique, the query complexity of the Schr\"odingerisation method is contingent upon $\mathcal{O}(\delta^{-3})$, but is {\it independent} of the minimal eigenvalue of the matrix $A$. This distinctive attribute implies that, in tackling specific types of partial differential equations (PDEs), the Schr\"odingerisation method can potentially match or even surpass the efficiency of general block-encoding methods. For instance, when applied to solve the heat equation, the Schr\"odingerisation method exhibits a comparable complexity to that of block-encoding.

In conclusion, the exploration of Nagy's theorem and its implications, coupled with the simple and generic Schr\"odingerisation approach, enhances our understanding of quantum dynamics and computation. These techniques offer promising avenues for addressing challenges posed by unitary and non-unitary operators, paving the way for more efficient and accurate quantum algorithms. As the field of quantum computing continues to evolve, the interplay between operator dilations, unitary representations, and specialized techniques like Schr\"odingerisation promises to shape the future of quantum computation and simulation. 

\section*{Acknowledgement}
SJ was partially supported by the NSFC grant No. 12031013, the Shanghai Municipal Science
and Technology Major Project (2021SHZDZX0102), and the Innovation Program of Shanghai Municipal Education Commission (No. 2021-01-07-00-02-E00087). NL acknowledges funding from the Science and Technology Program of Shanghai, China (21JC1402900).
LZ was partially supported by the NSFC grant No. 12271360, the Shanghai Municipal Science and Technology Project (22JC1401600).
SJ, NL and LZ are also supported by the Fundamental Research Funds for the Central Universities.

\bibliographystyle{plain}
\bibliography{ref.bib}

\appendix

\section{Nagy's dilation theorem}
\label{sec:appendix:nagy}
\begin{mydef}
Let $G$ be a grouop. 
\begin{enumerate}
    \item A mapping $V(t)$ on $G$, whose values are bounded operators on a Hilbert space $\mathcal{H}$, is said to be \textit{positive definite} if $V(s^{-1})=V(t)^\dagger $ for every $s\in G$, and 
    \begin{equation}
        \sum_{s\in G} \sum_{t\in G} \langle V(t^{-1}s)h(s), h(t)\rangle \geq 0
    \end{equation}
    for every finitely nonzero mapping $h(s)$ from $G$ to $\mathcal{H}$, that is, which has values different from $0$ on a finite subset of $G$ only.
    
    \item A mapping $U(t)$ on $G$ is  a \textit{unitary representation} of the group $G$ if its values are unitary operators on a Hilbert space $\mathfrak{K}$ and it satisfies the condition $U(e)=I$ ($e$ being the identity element of $G$) and $U(s)U(t)=U(st)$ for $s,t\in G$.
\end{enumerate}
\end{mydef}

\begin{mythm}[{\cite[Theorem 7.1]{nagy2010harmonic}}]
Let $V(t)$ be a mapping defined on the group $G$, whose values are operators on $\mathcal{H}$. Then $V(t)$ is positive definite with $T(e)=I_{\mathcal{H}}$ if and only if there exists $U(t)$, a unitary representation of $G$ on a space $\mathfrak{K}$ containing $\mathcal{H}$ as a subspace, such that
\begin{equation}
    V(t) = P_{\mathcal{H}} U(t),
\end{equation}
and 
\begin{equation}\label{eqn:append:nagy:min}
    \mathfrak{K} = \bigvee_{s\in G} U(t)\mathcal{H} \text{ (minimality condition)}.
\end{equation}
This unitary representation of $G$ is determined by the mapping $V(t)$ up to isomorphism so that one can call it ``the minimal unitary dilation'' of the mapping $V(t)$.
\end{mythm}

\begin{proof}
\textbf{Sufficiency}. Suppose $U(t)$ is a unitary representation of $G$, then
\begin{equation}
\begin{aligned}
    T(e) &= P_{\mathcal{H}}U(e) = I_{\mathcal{H}}, \\
    T(t^{-1}) &= P_{\mathcal{H}}U(t^{-1}) = P_{\mathcal{H}}U(t)^\dagger  = (P_{\mathcal{H}}U(t))^\dagger  = V(t)^\dagger ,
\end{aligned}
\end{equation}
and 
\begin{equation}
    \sum_{s\in G} \sum_{t\in G} \langle P_{\mathcal{H}}U(t^{-1}s)h(s), h(t)\rangle = \sum_{s\in G} \sum_{t\in G} \langle U(t)^\dagger U(t)h(s), h(t)\rangle = \left\| \sum_{s\in G} U(t)h(s) \right\|^2 \geq 0
\end{equation}
for every finitely nonzero mapping $h(s)$ from $G$ to $\mathcal{H}$.

\textbf{Necessity}. Consider the set $\mathcal{H}_1$, obviously linear, of the finitely nonzero mapping $h(s)$ from $G$ to $\mathcal{H}$, and  define on $\mathcal{H}_1$ a bilinear form by 
\begin{equation}
    \langle H,H^\prime \rangle = \sum_{s}\sum_{t}\langle T(t^{-1}s)H(s), H^\prime(t)\rangle, \quad [H=H(s),H^\prime=H^\prime(s)].
\end{equation}
Since $\langle H,H^\prime \rangle\geq 0$, the elements for which $\langle H,H^\prime \rangle=0$ form a linear manifold $\mathcal{N}$ in $\mathcal{H}_1$ by Schwarz's inequality. Therefore the form $\langle H,H^\prime \rangle$ defines in a natural way a bilinear form $(k,k^\prime)$ on the quotient space $\mathfrak{K}_0 = \mathcal{H}_1/\mathcal{N}$. The corresponding quadratic form $(k,k)$ is positive definite on $\mathfrak{K}_0$, and $\|k\| = (k,k)^{1/2}$ is a norm on $\mathfrak{K}_0$; by completing $\mathfrak{K}_0$ with respect to this norm one obtains a Hilbert space $\mathfrak{K}$.

Now we embed $\mathcal{H}$ in $\mathfrak{K}$ (and even in $\mathfrak{K}_0$) by identifying the element $h$ of $\mathcal{H}$ with the mapping $H = H(s) = \delta_e(s)h$ (where $\delta_e (e) = 1$ and $\delta_e (s) = 0$ for $s \neq e$), or, more precisely, with the equivalence class modulo $\mathcal{N}$ determined by this mapping. This identification is allowed because it preserves the linear and metric structure of $\mathcal{H}$. Indeed, one has
\begin{equation}
    \langle H, H^\prime \rangle =  \sum_{s}\sum_{t} \langle T(t^{-1}s)\delta_e(s)h, \delta_e(t)h^\prime \rangle  = \langle T(e)h, h^\prime \rangle  = \langle h,h^\prime \rangle .
\end{equation}

Now we set, for $H=H(s)\in \mathcal{H}_1$ and $a\in G$, 
\begin{equation}
    H_{a} = H(a^{-1}s).
\end{equation}
It is obvious that 
\begin{equation}
    \langle H_{a}, H^\prime_{a} \rangle = \sum_{s}\sum_{t} \langle T(t^{-1}s)H(a^{-1}s), H^\prime(a^{-1}t)\rangle = \sum_{\sigma}\sum_{\tau} \langle T(\tau^{-1}\sigma)H(\sigma), H^\prime(\tau)\rangle = \langle H,H^\prime \rangle.
\end{equation}
Therefore $H\in\mathcal{N}$ implies $H_{a}\in\mathcal{N}$ and consequently the transformation $H\rightarrow H_{a}$ in $H$ generates a transformation $k\rightarrow k_{a}$ of the equivalence classes modulo $\mathcal{N}$. Setting $U(a)k=k_{a}$, thus we define for every $a\in G$ a linear transformation of $\mathfrak{K}_0$ onto $\mathfrak{K}_0$, such that $U(e)=I$, $U(a)U(b)=U(ab)$ and $(U(a)k, U(a)k^\prime)=(k,k^\prime)$. These transformations on $\mathfrak{K}_0$ can be extended by continuity to unitary transformations on $\mathfrak{K}$, forming a representation of the group $G$.

For $H, H^\prime\in\mathcal{H}$ we obtain (setting $\delta_a(s)=\delta_e(a^{-1}s)$)
\begin{equation}
    \langle U(a)H, H^\prime \rangle_{\mathfrak{K}}  = \langle \delta_{a}h, \delta_{e}h^\prime \rangle = \sum_{s}\sum_{t} \langle T(t^{-1}s)\delta_{a}(s)h, \delta_{e}(t)h^\prime \rangle  = \langle T(a)h, h^\prime \rangle ,
\end{equation}
and hence
\begin{equation}
    T(a) = P_{\mathcal{H}}U(a) \quad \text{ for every } a \in G.
\end{equation}

\textbf{Minimality}. Observe  that every mapping $H = H(s) \in \mathcal{H}_1$ can be considered as a finite sum of terms of the type $\delta_{\sigma}(s)h$ (i.e., the type $(\delta_e(s)h)_{\sigma}$ $(\sigma \in G)$), and hence every element $k$ of $\mathfrak{K}_0$ can be decomposed into a finite sum of terms of the type $U (\sigma)h$ $(\sigma\in G, h\in \mathcal{H})$. This implies \eqref{eqn:append:nagy:min}.
\end{proof}

\section{Lemmas of Block-encoding}
\label{sec:lemmas}
Here we give some known results about blockencoding.

\begin{mylemma}[{\cite[Lemma 52]{gilyen2019quantum}}]\label{lemma:block:lcu}
Let $A=\sum_{j=1}^{N} y_jA_j$ be an $m$-qubit operator, $\log_2(N)$ be an integer and $\delta\in\mathbb{R}_{+}$. Suppose that $O_{coef}$ is an oracle for the LCU coefficients $y_j$, $W=\sum_{j=0}^{N-1}|j\rangle \langle j|\otimes U_j$ is an $m+a+\log_2(N)$-qubit unitary such that for all $j\in\{0,\dots,N\}$, $U_j$ is an $(\alpha,a,\delta_1)$-block-encoding of $A_j$. Then one can implement a $(\alpha, a+\log_2(N), \alpha\delta_1)$-block-encoding of $A$, with a single use of $W$ and $O_{coef}$.
\end{mylemma}

\begin{mylemma}[{\cite[Lemma 53]{gilyen2019quantum}}]\label{lemma:block:multiply}
If $U$ is an $(\alpha,a,\delta)$-block-encoding of an $m$-qubit operator $A$, and $V$ is an $(\beta,b,\varepsilon)$-block-encoding of an $m$-qubit operator $B$,  then $(I_b \otimes U )(I_a \otimes V )$ is an $(\alpha\beta, a + b,\alpha\varepsilon +\beta\delta)$-block-encoding of $AB$.
\end{mylemma}

\begin{mylemma}[{\cite[Corollary 60]{gilyen2019quantum}}]\label{lemma:block:expiHt}
    Let $\delta \in (0, \frac{1}{2} ), t \in \mathbb{R}$ and $\alpha \in \mathbb{R}_+$. Let $U$ be an $(\alpha,a,0)$-block-encoding of the unknown Hamiltonian $H$. In order to implement an $\delta$-precise Hamiltonian simulation unitary $V$ which is an $(1,a + 2,\delta)$-block-encoding of $e^{itH}$, it is necessary and sufficient to use the unitary $U$ a total number of times
    \begin{equation}
        \mathcal{O}\left( \alpha|t| + \frac{\log(1/\varepsilon)}{\log(e+\log(1/\delta)/(\alpha|t|))} \right).
    \end{equation}
\end{mylemma}

\begin{mylemma}[{\cite[Lemma 20]{an2022theory}}]\label{lemma:block:expAt}
Consider solving Equation \eqref{mod:origin} where $A$ is a positive definite Hermitian matrix with $\|A\|_2 \leq 1$ and all the eigenvalues of $A$ are within the interval $[\lambda_{\min}, 1]$ for a $\lambda_{\min} > 0$. Suppose that a $(1,l,0)$-block-encoding of $A$ is given, denoted by $U_A$. Then for any $T > 0$ and $\delta < 1/4$, a $(3,l+4,\delta)$-block-encoding of $e^{-AT}$ can be constructed using
\begin{equation}
    \mathcal{O}\left( \frac{\sqrt{\tilde{T}}}{\lambda_{\min}} \log\left( \frac{1}{\delta} \right) \log\left( \frac{\tilde{T}\log(1/\delta)}{\delta} \right) \right)
\end{equation}
queries to $U_A$, its inverse and controlled versions, and
\begin{equation}
    \mathcal{O}\left( \left( l+\frac{1}{\lambda_{\min}}\log\left( \frac{\tilde{T}\log(1/\delta)}{\delta} \right) \right) \sqrt{\tilde{T}} \log \left( \frac{1}{\delta} \right) \right)
\end{equation}
additional one- and two- qubit gates with $\tilde{T}=\max(T,\log(1/\delta))$.
\end{mylemma}

\section{Proof of equation \eqref{eqn:schro:cv:projU}}
\label{sec:proof:residue}
For $h_0, h_0^\prime \in \mathcal{H}$, $t>0$, we would like to calculate
\begin{equation}\label{eqn:append:nagy:timeUa}
    \langle U(t)h, h^\prime \rangle = \frac{1}{2\pi} \int_{\mathbb{R}} \frac{2}{\eta^2+1} \left\langle e^{-i(\eta H_1+H_2)t}h_0, h_0^\prime \right \rangle  \dif \eta \triangleq \frac{1}{2\pi} \int_{\mathbb{R}} g(\eta) \dif \eta.
\end{equation}
Note that {\it{each entry}} of the integrand on the right hand side of \eqref{eqn:append:nagy:timeUa} is analytic with respect to $\eta\neq \pm i$. We define the integral along the half circle with a radius $R$ in the lower half plane $\gamma = [-R,R] \cup \gamma_0 = [-R,R] \cup \{ Re^{i\theta}|\theta\in[-\pi,0] \}$, then from Cauchy's residue theorem,
\begin{equation}
    \frac{1}{2\pi} \int_{\gamma} g(\eta) \dif \eta = i \text{Res}(g, -i) = -\left\langle \mathcal{T} e^{- (H_1+iH_2)t}h_0, h_0^\prime \right \rangle  = -\langle V(t)h_0, h_0^\prime \rangle .
\end{equation}
Assumethat the eigenvalues of $H_1(t)$ are uniformly bounded from below by $\lambda_0 = \lambda_{\min}(H_1) \geq 0$, then
\begin{equation}
    \left|\int_{\gamma_0} g(\eta) \dif \eta\right| \leq \int_{-\pi}^{0} \frac{2}{\left|R^2e^{2i\theta}+1\right|} e^{ \lambda_0 Rt\sin\theta } |h_0||h_0^\prime| |Re^{i\theta}| \dif \theta \leq 2\pi \frac{R}{R^2-1} |h_0||h_0^\prime|,
\end{equation}
hence
\begin{equation}
    \lim_{R\rightarrow +\infty} \frac{1}{2\pi} \int_{\gamma_0} g(\eta) \dif \eta = 0,
\end{equation}
\begin{equation}\label{eqn:append:nagy:timeUa2}
    \frac{1}{2\pi} \int_{\mathbb{R}} g(\eta) \dif \eta = \lim_{R\rightarrow +\infty} \frac{1}{2\pi} \int_{[-R,R]} g(\eta) \dif \eta = -\lim_{R\rightarrow +\infty} \frac{1}{2\pi} \int_{\gamma} g(\eta) \dif \eta = \langle V(t)h_0, h_0^\prime \rangle .
\end{equation}
Substituting \eqref{eqn:append:nagy:timeUa2} into \eqref{eqn:append:nagy:timeUa} gives
\begin{equation}
    \langle U(t)h, h^\prime \rangle = \langle V(t) h_0, h_0^\prime \rangle, \quad \forall t \in \mathbb{R}, 
\end{equation}
where $h=h(\eta)\equiv h_0, h^\prime=h^\prime(\eta)\equiv h_0^\prime$. That is, $V(t)$ is the projection of the unitary operator $U(t)$ in $\mathcal{H}$, explicitly,
\begin{equation}
    V(t) = P_{\mathcal{H}} U(t) = \frac{1}{2\pi} \int_{\mathbb{R}} \frac{2}{\eta^2+1} U(t) d\eta.
\end{equation}

\section{Proof of Theorem \ref{thm:block:expAt:nonHer}}
\label{sec:proof:block:expAt:nonHer}

\begin{proof}
As described in Section \ref{sec:block}, we first make the decomposition $A=H_1 + iH_2$ in terms of Hermitian matrices, then apply the first-order Lie-Trotter formula with $\Delta t=T/K$,
\begin{equation}\label{eqn:proof:block:K}
    \mathcal{S}_1(T) = (\mathcal{S}_1(\Delta t))^K := \left( e^{-H_1T/K} e^{-iH_2T/K} \right)^K,
\end{equation}
where $K$ is large enough to approximate $e^{-AT}$. Through a similar analysis as \cite[Lemma1]{childs2021theory}, the Trotter error with $1$-norm scaling is $\mathcal{O} \left( \left( \|H_1\|_2 + \|H_2\|_2 \right)^2T^2/K \right)$ and it suffices to choose
\begin{equation}
    K = \mathcal{O} \left( \left( \|H_1\|_2 + \|H_2\|_2 \right)^2T^2/\delta^\prime \right) = \mathcal{O} \left( (\|A\|_2 T)^2/\delta^\prime \right) = \mathcal{O} \left( (s(A)\|A\|_{\max}T)^2/\delta^\prime \right)
\end{equation}
to simulate with accuracy $\delta^\prime$, where we have used the fact $\|A\|_2\leq s(A)\|A\|_{\max}$.

Since $H_1$ is Hermitian, we can construct a $(1,m+2,0)$-block-encoding $U_{H_1}$ of $H_1/(s(H_1)\|H_1\|_{\max})$ by Lemma \ref{lemma:block:hamiltonian} using $\mathcal{O}(1)$ queries. Denoting 
\begin{equation}
    T^\prime = s(H_1)\|H_1\|_{\max}T/K, \quad  H_1^\prime=H_1/(s(H_1)\|H_1\|_{\max}),
\end{equation}
then $\|H_1^\prime\|_2 \leq 1$, Lemma \ref{lemma:block:expAt} indicates that a $(3,m+6,\delta_1)$-block-encoding $U_1$ of $e^{-H_1T/K}=e^{-H_1^\prime T^\prime}$ can be constructed using
\begin{equation}\label{eqn:proof:block:Q1}
    Q_1 = \mathcal{O}\left( \frac{\sqrt{\tilde{T}^\prime}}{\lambda_{\min}(H_1^\prime)} \log\left( \frac{1}{\delta_1} \right) \log\left( \frac{\tilde{T}^\prime \log(1/\delta_1)}{\delta_1} \right) \right) = \mathcal{O}\left( \frac{s(A)\|A\|_{\max}}{\lambda_{0}(A)} \log^{1.5}\left( \frac{1}{\delta_1} \right) \log\left( \frac{\log(1/\delta_1)}{\delta_1} \right) \right)
\end{equation}
queries to $U_{H_1}$, its inverse and controlled versions, and
\begin{equation}
    G_1 = \mathcal{O}\left( \left(m + \frac{s(H_1)\|H_1\|_{\max}}{\lambda_{\min}(H_1)} \log\left( \frac{\log(1/\delta_1)}{\delta_1} \right) \right) \log^{1.5}\left( \frac{1}{\delta_1} \right) \right)
\end{equation}
additional one- and two- qubit gates.

Similarly, we can construct a $(1,m+2,0)$-block-encoding $U_{H_2}$ of $H_2/(s(H_2)\|H_2\|_{\max})$ using $\mathcal{O}(1)$ queries. Denoting 
\begin{equation}
    T^\prime = s(H_2)\|H_2\|_{\max}T/K, \quad  H_2^\prime=H_2/(s(H_2)\|H_2\|_{\max}),
\end{equation}
from the known complexity results of block-Hamiltonian simulation by Lemma \ref{lemma:block:expiHt}, in order to implement a $(1,m+4,\delta_2)$-block-encoding $U_2$ of $e^{-iH_2T/K}=e^{-iH_2^\prime T^\prime}$, it is necessary and sufficient to use a total number of
\begin{equation}\label{eqn:proof:block:Q2}
    Q_2 = \mathcal{O} \left( T^\prime + \frac{\log(1/\delta_2)}{\log(e+\log(1/\delta_2)/T^\prime)} \right) = \mathcal{O}\left( \frac{\log(1/\delta_2)}{\log(\tau/\delta^\prime) + \log\log(1/\delta_2)} \right)
\end{equation}
the unitary $U_{H_2}$, 3 controlled-U or its inverse, $\mathcal{O}(1)$ ancilla qubits and
\begin{equation}
    G_2 = \mathcal{O}\left( \frac{m \log(1/\delta_2)}{\log(\tau/\delta^\prime) + \log\log(1/\delta_2)}\right) 
\end{equation}
two-qubit gates with $\tau=s(A)\|A\|_{\max}T$. Note that $K$ is chosen sufficiently large such that $T^\prime < 1$. 

Given $U_1$ a $(3,m+6,\delta_1)$-block-encoding of $e^{-H_1T/K}$, $U_2$ a $(1,m+4,\delta_2)$-block-encoding of $e^{-iH_2T/K}$, we know $\tilde{U} = (I_{m+4}\otimes U_1)(I_{m+6}\otimes U_2)$ is a $(3, 2m+10, 3\delta_2+\delta_1)$ of $e^{-H_1T/K}e^{-iH_2T/K}$ by Lemma \ref{lemma:block:multiply}. To bound the overall error for time $T/K$ by $\delta^\prime/K>0$, it suffices to choose $\delta_1=\delta_2 = \delta^\prime/(4K) = (\delta^\prime)^2/(4\tau^2)$, therefore the construction of $\tilde{U}$ requires $(Q_1+Q_2)$ queries and $(G_1+G_2)$ two-qubit gates.

The complexity of solving equation \eqref{mod:origin} that outputs an $\delta$-approximation of $|u(T)\rangle$ with $\Omega(1)$ success probability is the number of segments $K$ times the query complexity for each segment $Q_1+Q_2$ (or gate complexity $G_1+G_2$) times the number of steps needed for oblivious amplitude amplification $\mathcal{O}\left( \|u(0)\|/\|u(T)\| \right)$. To achieve the precision $\delta$, we need to take $\delta^\prime = \mathcal{O}(\delta \left\|u(T)\right\|/\left\|u(0)\right\|)$, then we complete the proof by combining \eqref{eqn:proof:block:K}, \eqref{eqn:proof:block:Q1} and \eqref{eqn:proof:block:Q2}.

\end{proof}

\section{Proof of Theorem \ref{thm:schro:dv:expAt}}
\label{sec:proof:schro:expAt}

\begin{proof}
As described in Section \ref{sec:schro:dv}, we apply the first-order Lie-Trotter formula to approximate $U_j$
\begin{equation}
    U_{j}(T) \approx \left( e^{-i \eta_j H_1 T/K} e^{-i H_2 T/K} \right)^K
\end{equation}
with $K$ large enough. Through a similar analysis of the Trotter error as Appendix \ref{sec:proof:block:expAt:nonHer}, it suffices to choose
\begin{equation}\label{eqn:proof:schro:K}
    K = \mathcal{O} \left( \left( \|\eta_j H_1\|_2 + \|H_2\|_2 \right)^2T^2/\delta^\prime \right) = \mathcal{O} \left( (\|A\|_2 LT)^2/\delta^\prime \right) = \mathcal{O} \left( (s(A)\|A\|_{\max} LT)^2/\delta^\prime \right)
\end{equation}
to simulate with accuracy $\delta^\prime$, where we have used the fact $\|A\|_2\leq s(A)\|A\|_{\max}$ and $|\eta_j|\leq L$. 

Since $H_1$ is Hermitian, we can construct a $(1,m+2,0)$-block-encoding $U_{H_1}$ of $H_1/(s(H_1)\|H_1\|_{\max})$ using $\mathcal{O}(1)$ queries. Denoting 
\begin{equation}
    T^\prime = s(H_1)\|H_1\|_{\max} \eta_j T /K, \quad  H_1^\prime=H_1/(s(H_1)\|H_1\|_{\max}),
\end{equation}
from the known complexity results of block-Hamiltonian simulation by Lemma \ref{lemma:block:expiHt}, in order to implement a $(1,m+4,\delta_1)$-block-encoding $U_{\eta_j}$ of $e^{-i \eta_j H_1 T/K}=e^{-iH_1^\prime T^\prime}$, it is sufficient to use a total number of
\begin{equation}\label{eqn:proof:schro:Q}
    Q_1 = \mathcal{O} \left( T^\prime + \frac{\log(1/\delta_1)}{\log(e+\log(1/\delta_1)/T^\prime)} \right) = \mathcal{O}\left( \frac{\log(1/\delta_1)}{\log(L\tau/(\delta^\prime)) + \log\log(1/\delta_1)} \right)
\end{equation}
the unitary $U_{H_1}$, 3 controlled-$U_{H_1}$ or its inverse, $\mathcal{O}(1)$ ancilla qubits and
\begin{equation}
    G_1 = \mathcal{O}\left( \frac{m \log(1/\delta_1)}{\log(L\tau/(\delta^\prime)) + \log\log(1/\delta_1)} \right)
\end{equation}
two-qubit gates with $\tau=s(A)\|A\|_{\max}T$. To apply the LCU technique, we need to construct the select oracle
\begin{equation}
    \text{select}(\vec{U}) := \sum_{j=1}^{N} |j\rangle \langle j| \otimes  e^{-i \eta_j H_1 T/K},
\end{equation}
which can be written as
\begin{equation}
    \text{select}(\vec{U}) = \sum_{j=1}^{N} |j\rangle \langle j| \otimes  e^{-i (-L+2Lj/N) H_1 T/K} = e^{iLH_1T/K} \sum_{j=1}^{N} |j\rangle \langle j| \otimes  \left(e^{-i (2L/N) H_1 T/K}\right)^{j}.
\end{equation}
According to \cite[Lemma 6]{an2023linear}, the second operator can be constructed with cost $\mathcal{O}(\log(N))$ of $U_{\eta_j}$s. Therefore, the total query complexity of constructing $\text{select}(\vec{U})$ is $\mathcal{O}(Q_1\log(N))$. The complexity of constructing a $(1,m+4,\delta_2)$-blcok-encoding $U_{2}^\prime$ of $e^{-iH_2T/K}$ is similar as those in Appendix \ref{sec:proof:block:expAt:nonHer}.

Given $\text{select}(\vec{U})$, $U_{H_2}$ and $O_{coef}$ the LCU coefficient oracle, we first apply $O_{coef}\otimes O_{prep}$, then sequentially apply $\text{select}(\vec{U})$ and $U_{2}^\prime$ for $K$ times, finally apply $O_{coef}^\dagger$ on the ancilla register and evaluate it via amplitude estimation. The complexity of solving equation \eqref{mod:origin} that outputs an $\delta$-approximation of $|u(T)\rangle$ with $\Omega(1)$ success probability is the number of Trotter decomposition segments $K$ times the query complexity for each segment $\mathcal{O}(Q_1\log_2(N)+Q_2)$ times the number of steps needed for oblivious amplitude amplification $\mathcal{O}\left(\|y\| \|u(0)\|/\|u(T)\| \right)$. To achieve the precision $\delta$, one needs to take
\begin{equation}\label{eqn:proof:schro:param}
    \delta^\prime = \mathcal{O}\left(\frac{\delta \|u(T)\|}{\|u(0)\|}\right), \quad \delta_1 = \delta_2 = \mathcal{O}\left(\frac{\delta^\prime}{K}\right), \quad L = \mathcal{O}\left( \frac{1}{\delta^\prime}\right), \quad N = \frac{2L}{\Delta\eta} = \mathcal{O}\left(\frac{\tau}{(\delta^\prime)^2}\right).
\end{equation}
The estimate of $N$ comes from the global error of discretising $\eta$ to $N$ modes $\eta_j, j=1,\dots,N$. Combining \eqref{eqn:proof:schro:param}, \eqref{eqn:proof:schro:K} and \eqref{eqn:proof:schro:Q}, we obtain the total query complexity
\begin{equation}
    \tilde{\mathcal{O}} \left(  \left(\frac{\left\|u(0)\right\|}{\left\|u(T)\right\|}\right)^{4} \frac{(s(A)\|A\|_{\max}T)^2}{\delta^{3}} \right), 
\end{equation}
where $\tilde{\mathcal{O}}$ denotes $\mathcal{O}$ ignoring logarithmic terms. The gate complexity can be obtained similarly.
\end{proof}

\end{document}